\documentclass[conference]{IEEEtran}
\IEEEoverridecommandlockouts
\usepackage{cite}
\usepackage{amsmath,amssymb,amsfonts,amsthm}
\usepackage{algorithmicx,algorithm,algpseudocode}
\usepackage{graphicx}
\usepackage{textcomp}
\usepackage{xcolor}
\usepackage{subfigure}
\def\BibTeX{{\rm B\kern-.05em{\sc i\kern-.025em b}\kern-.08em
    T\kern-.1667em\lower.7ex\hbox{E}\kern-.125emX}}
    
    \newtheorem{prop}{Proposition}
\begin{document}

\title{Sparse Millimeter Wave Channel Estimation From Partially Coherent Measurements
}

 \author{Weijia Yi$^{\dagger}$, Nitin Jonathan Myers$^{\ast}$, and Geethu Joseph$^{\dagger}$ \\
 			$^{\dagger}$Department of Microelectronics, Delft University of Technology, The Netherlands\\
		$^{\ast}$Delft Center for Systems and Control, Delft University of Technology, The Netherlands\\
		 Emails: $\{$w.yi-1@student.tudelft.nl\},\{N.J.Myers, G.Joseph$\}$@tudelft.nl}
	\maketitle

\maketitle

\begin{abstract}
This paper develops a channel estimation technique for millimeter wave (mmWave) communication systems. Our method exploits the sparse structure in mmWave channels for low training overhead and accounts for the phase errors in the channel measurements due to phase noise at the oscillator. Specifically, in IEEE 802.11ad/ay-based mmWave systems, the phase errors within a beam refinement protocol packet are almost the same, while the errors across different packets are substantially different. Consequently, standard sparsity-aware algorithms, which ignore phase errors, fail when channel measurements are acquired over multiple beam refinement protocol packets. We present a novel algorithm called partially coherent matching pursuit for sparse channel estimation under practical phase noise perturbations. Our method iteratively detects the support of sparse signal and employs alternating minimization to jointly estimate the signal and the phase errors. We numerically show that our algorithm can reconstruct the channel accurately at a lower complexity than the benchmarks. 
\end{abstract}


\section{Introduction}
Millimeter wave (mmWave) systems, currently used in 5G and IEEE 802.11ad/ay devices, are capable of achieving Gbps data rates by beamforming over wide bandwidths. Unlike the lower frequencies, the propagation characteristics at mmWave result in high scattering~\cite{9568459}. So, the wireless channel between the transmitter and the receiver is sparse in the angle domain representation. This sparse structure allows the use of compressed sensing (CS) algorithms for channel estimation from fewer measurements compared to classical channel estimation, thereby reducing the training overhead~\cite{9521836}.  Unfortunately, the phase of the compressed channel measurements acquired with typical IEEE 802.11ad/ay hardware is perturbed. The phase perturbations occur due to residual carrier frequency offset as well as random phase noise at the oscillator~\cite{7835110}, which is more significant at mmWave than lower carrier frequencies
\cite{9103070}. The phase errors induced by phase noise result in a model mismatch between the standard CS model and the observed measurement model. Due to this mismatch, standard CS-based channel estimation methods \cite{7913633, alkhateeb2014channel, 8306126} that are agnostic to such phase errors fail. Therefore, in this paper, we focus on estimating sparse mmWave channels utilizing the sparse structure while handling the phase errors in the measurements.

\par We next briefly review the existing sparse mmWave channel estimation algorithms handling phase error. One of the early studies combined a CS algorithm with a Kalman filter to track the phase noise, but a large phase noise variance could invalidate the method~\cite{7905916}. Several other works considered phaseless measurements with independent random phase offset on each training slot, forming a joint phase retrieval and sparse recovery problem. One such approach used gradient descent for phase retrieval but required several measurements with known phase errors for calibration~\cite{10012988}. Also, \cite{ling2015self} jointly estimated the phase errors and the channel by constructing high-dimensional sparse tensors, resulting in high computational complexity. To decrease the complexity, a sparse bipartite graph code-based phaseless decoding scheme was introduced~\cite {8770109}. The above algorithms in~\cite{10012988,ling2015self,8770109}, however, do not exploit the specific structure of phase errors in the measurements owing to the short packet signaling used for beam training in IEEE 802.11ad/ay standards. 
To handle this structure, \cite{9007529} developed a partially coherent compressive phase retrieval (PC-CPR) technique to solve for sparse vectors from a \emph{partially coherent measurement model}. Here, partially coherent indicates that the phase errors fluctuate slightly during a short time slot, but are relatively independent across packets. Finally, \cite{9801831} developed a message passing-based algorithm for partially coherent sparse recovery. Such an approach usually requires higher computational complexity than greedy matching pursuit-based algorithms.
\par We present a compressive greedy algorithm for sparse channel estimation with a partially coherent measurement model. Such a model is motivated by the signaling structure used in IEEE 802.11ad/ay standards, wherein a spatial channel measurement is acquired using the training (TRN) subfield within a beam refinement protocol (BRP) packet~\cite{9502043}. With a typical phase noise process, phase coherence is preserved within a packet, but lost across different packets. 
Our main contribution is the development of a joint support detection rule across packets, combining alternating minimization and matching pursuit techniques to estimate both the channel and phase noise for channel estimation. To assess our algorithm's performance, we also derive guarantees on partial support recovery. By assuming that the phase noise remains constant within a BRP packet but varies over packets, we substantially reduced the number of parameters to be estimated. Also, matching pursuit in our method results in a lower computational complexity than other benchmarks. 
Finally, we show by simulations that our algorithm
outperforms orthogonal matching pursuit (OMP), PC-CPR, and Sparse-Lift~\cite{ling2015self} in the presence of phase noise.


\section{Channel model and frame structure}
\label{sec:model}
We consider a mmWave system with an analog antenna array comprising $N$ antennas at the transmitter (TX) and a single antenna receiver (RX) as shown in Fig.~\ref{fig:system}. The focus of this paper is on transmit beam alignment through channel estimation. Although we assume a single antenna at the RX for ease of exposition, our approach can be extended to multi-antenna receivers using an appropriate array response vector. 
\begin{figure}[htbp]
\centerline{\includegraphics[width=0.9\columnwidth]{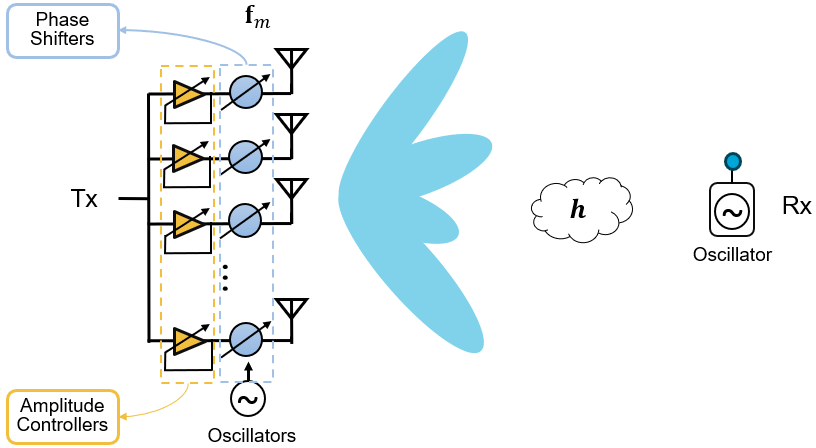}}
\caption{An mmWave MISO system with an analog array at the TX and a single antenna RX. Our goal is to estimate the channel $\mathbf{h}$ under phase noise.}
\label{fig:system}
\end{figure}

We consider a narrowband system and model the multiple input single output (MISO) channel between the TX and the RX as a vector $\mathbf{h} \in \mathbb{C}^{N}$.  Let $K$ denote the number of propagation paths in the environment. The path gain and the direction of arrival associated with the $k^{\mathrm{th}}$ path are denoted by $h_k$ and $\theta_k$. Then, the channel is
\begin{equation}
\label{eq:channel_spatial}
    \mathbf{h} = \sum^K_{k=1}h_k\mathbf{a}(\theta_k).
\end{equation}
Here, $\mathbf{a}(\theta)$ is the transmit array response vector given by
\begin{equation}
    \mathbf{a}(\theta) = \left[1, e^{\mathsf{j}\pi\sin{\theta}},\cdots , e^{\mathsf{j}\pi(N-1)\sin{\theta}}\right]^{\mathsf{T}},
\end{equation}
for a half-wavelength spaced uniform linear array at the TX.
We use the discrete Fourier transform dictionary to obtain the angle-domain representation of the channel, 
\begin{equation}
\label{eq:channel_angle}
    \mathbf{x} = \mathbf{U}_N\mathbf{h},
\end{equation}
where $\mathbf{U}_N$ is the $N\times N$ unitary discrete Fourier transform matrix. The vector $\mathbf{x}$ in \eqref{eq:channel_angle} is sparse due to high scattering at mmWave carrier frequencies. 

We next describe the frame structure used to obtain measurements for channel estimation. To this end, we consider the IEEE 802.11ad frame structure shown in Fig.~\ref{fig:packet}. The TX applies distinct beamformers over $P$ different packets, for the RX to acquire channel measurements, within the channel coherence time. We define $M$ as the number of beamformers applied in each packet. With IEEE 802.11ad, $M$ can be at most $128$ \cite{standard}; however, it can often be smaller (e.g. $16$) to acquire redundant measurements and enhance the spreading gain. In that case, numerous beam refinement protocol (BRP) packets can be used to acquire enough spatial measurements for channel estimation.
\begin{figure}[htbp]
\centerline{\includegraphics[width=1\columnwidth]{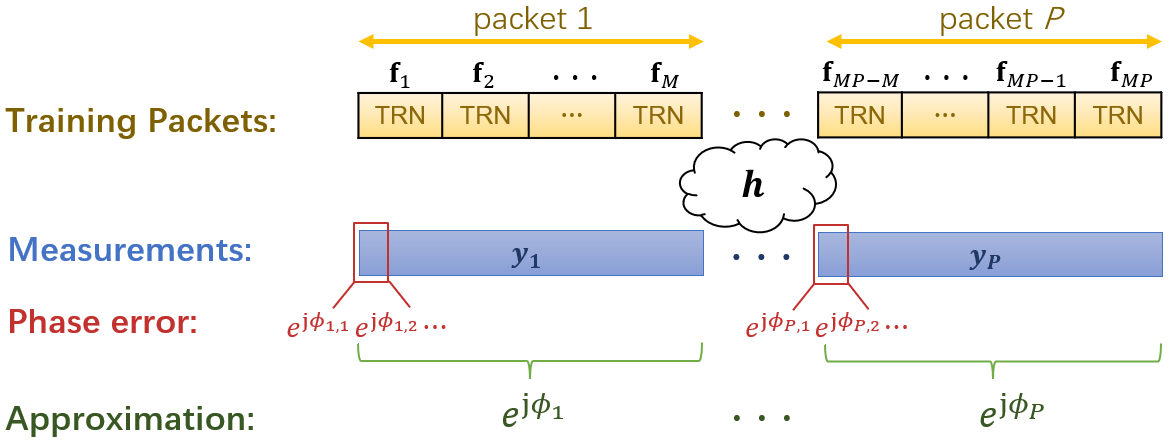}}
\caption{Partially coherent measurement model where the phase errors in the channel measurements are assumed to be constant within each packet.}
\label{fig:packet}
\end{figure}

Let $\mathbf{f}_{m}\in\mathbb{C}^{N}$ denote the $m^{\mathrm{th}}$ beamformer at the TX and $y[m]$ denote the $m^{\mathrm{th}}$ received channel measurement. This measurement is perturbed by phase noise $\phi_m$ and additive white Gaussian noise $w[m]$ of variance $\sigma^2$, i.e., 
\begin{equation}
y[m]=e^{\mathsf{j}\phi_m}\mathbf{f}_{m}^{\ast}\mathbf{h}+w[m]
\label{eq:sparse_phase_error}
=e^{\mathsf{j}\phi_m}\mathbf{f}_{m}^{\ast}\mathbf{U}^{\ast}_N \mathbf{x}+w[m],
\end{equation}
using \eqref{eq:channel_angle}. Note that $\mathbf{f}_{m}^{\ast}$ is the conjugate transpose of $\mathbf{f}_{m}$. Here, the phase noise $\phi_{m}$ can be modeled as a Wiener process, i.e., $ \phi_{m}|_{\phi_{m-1}} \sim \mathcal{N}(\phi_{m-1},\tau^2)$, where $\tau$ is $2\pi f_c\sqrt{c(T_m-T_{m-1})}$. Here, $f_c$ is the carrier frequency, $c$ is an oscillator-dependent constant, and $T_m$ is the time stamp associated with the $m^{\mathrm{th}}$ measurement. Under these modeling assumptions, we aim to estimate the sparse channel $\mathbf{x}$ using the measurements from \eqref{eq:sparse_phase_error} with the knowledge of $\mathbf{f}_{m}$, for $m=1,2,\ldots,M$. 




\section{Sparse Channel Estimation Algorithm}
In this section, we first reformulate the channel estimation problem into a CS problem to exploit the sparsity in $\mathbf{x}$ and then, develop a greedy algorithm for channel estimation. 
\subsection{Partially Coherent CS model}
\par To formulate the CS model, we note that the time difference between successive measurements in a BRP packet, i.e., $T_m-T_{m-1}$, is $128 \, \mathrm{ns}$ in IEEE 802.11ad. In contrast, the difference between the successive packets can range from $3\, \mu $s to $44\, \mu$s~\cite{kome2016further}. Therefore, a high variance phase offset is introduced in the measurements when switching to a new packet. Fig. \ref{fig:phase} shows a realization of $\phi_m$ when $M=16$ measurements are acquired in each of the $P=4$ packets. 
\begin{figure}[htbp]
\centerline{\includegraphics[trim=0.8cm 0.5cm 0.8cm 0.55cm,width=0.8\columnwidth]{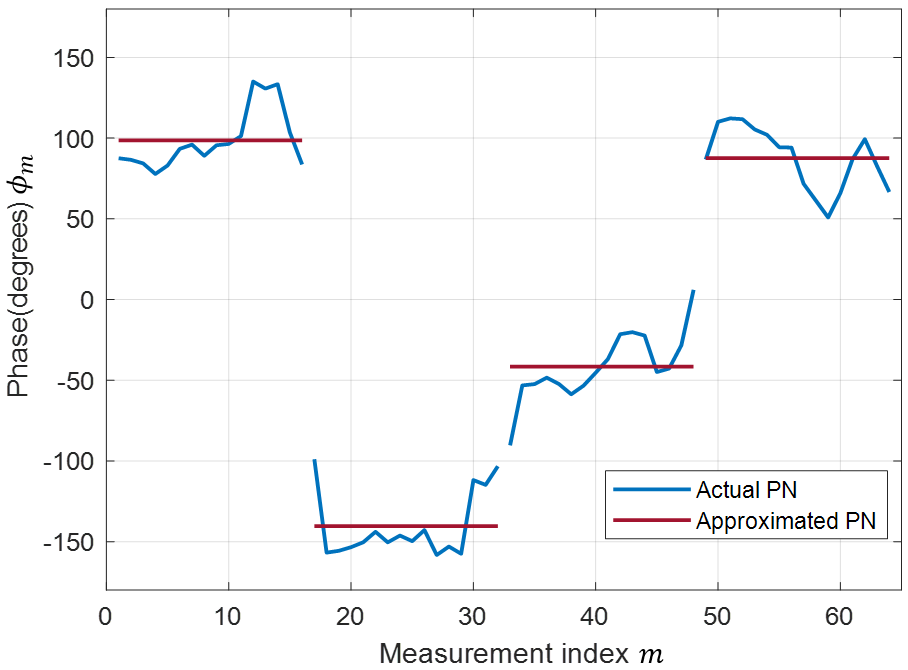}}
\caption{A realization of the phase noise process with the measurement index. Our model assumes that the phase error within each packet is the same.}
\label{fig:phase}
\end{figure}
To develop a tractable algorithm, we ignore the phase variations within each packet and only consider phase offsets across different packets. In Sec.~\ref{sec:simulations}, however, we evaluate our algorithm by also considering phase variations within the packets. 

\par Under the above simplifying assumption, the measurements can be expressed as a partially coherent CS model \cite{9007529}. We define $\phi_p$ as the phase error in the measurements acquired within the $p^{\mathrm{th}}$ packet. The vector of $MP$ channel measurements, acquired over $P$ packets, can be expressed in terms of the CS matrix $\mathbf{A}\in \mathbb{C}^{MP \times N}$ and the phase errors $\{\phi_p\}_{p=1}^{P}$. Based on \eqref{eq:sparse_phase_error}, we define the $m^{\mathrm{th}}$ row of the CS matrix as 
\begin{equation}
\label{eq:CS_matrix}
    \mathbf{A}(m,:)=\mathbf{f}^{\ast}_m \mathbf{U}^{\ast}_{N},
\end{equation}
and a diagonal matrix containing the phase errors as 
$
    \mathbf{\Phi}=\operatorname{diag}\left\{e^{\mathsf{j}\phi_{1}}\mathbf{1}_M,...,e^{\mathsf{j}\phi_{P}}\mathbf{1}_M \right\},
$
where $\mathbf{1}_M$ is a row vector comprising $M$ ones. The vector version of \eqref{eq:sparse_phase_error} is then
\begin{equation}\label{eq:measurement_modelMP}
    \mathbf{y}=\mathbf{\Phi}\mathbf{Ax}+\mathbf{w},
\end{equation}
where $\mathbf{\Phi}$ and $\mathbf{x}$ are unknown. We observe from \eqref{eq:measurement_modelMP} that $\mathbf{\mathbf{x}}$ can only be estimated up to a global phase. This is because $(\Phi_1 e^{-\mathsf{j}\delta}, \mathbf{x}_1 e^{\mathsf{j}\delta})$ is a solution to \eqref{eq:measurement_modelMP} whenever $(\Phi_1, \mathbf{x}_1)$ is a solution. Hence, the goal of partially coherent CS is to estimate $\mathbf{x}$, up to a global phase, from the measurements in \eqref{eq:measurement_modelMP}.

\par We can split the measurement model in \eqref{eq:measurement_modelMP} on a packet-by-packet basis. The measurements from the $p^{\mathrm{th}}$ packet correspond to rows $(p-1)M+1$ to $pM$ of \eqref{eq:measurement_modelMP}. We define the measurements acquired in the $p^{\mathrm{th}}$ packet as $\mathbf{y}_p$, and the CS measurement matrix associated with the $p^{\mathrm{th}}$ packet as $\mathbf{A}_p$. Specifically, $\mathbf{A}_p=\mathbf{A}((p-1)M+1:Mp,:)$ and 
\begin{equation}
\label{eq:measurement_model_pk}
    \mathbf{y}_p=e^{\mathsf{j}\phi_p}\mathbf{A}_p\mathbf{x}+\mathbf{w}_p,
\end{equation}
where $\mathbf{w}_p$ is additive noise. Now, our channel estimation problem is equivalent to estimating $\mathbf{x}$ from the phase perturbed measurements $\{\mathbf{y}_p\}_{p=1}^P$ acquired using CS matrices $\{\mathbf{A}_p\}_{p=1}^P$.

\subsection{Partially coherent matching pursuit (PC-MP)}
\label{sec:PC-MP}
\par Our PC-MP algorithm is a greedy approach that adds one element to the estimated support set in each iteration. Then, the algorithm estimates the phase errors and the sparse signal over the estimated support through alternating minimization. This procedure is carried out until the stopping criterion is met. In this paper, we use the assumption in \cite{9007529} that the number of non-zero coefficients in $\mathbf{x}$ is known.

\par We first discuss our support detection rule in PC-MP. Our algorithm is initialized by setting $\Lambda_0$ to an empty set and $\hat{\mathbf{x}}_{0}$ to a zero vector. Here, we use $\Lambda_t$ to denote the estimated support set of the sparse vector, $\hat{\mathbf{x}}_{t}$ as the estimate of $\mathbf{x}$, and $\hat{\phi}_{p,t}$ as the phase error estimated for packet $p$, after $t$ iterations. The vector of estimated phase errors is denoted by $\hat{\boldsymbol{\phi}}_t$. For the $p^{\mathrm{th}}$ packet, we denote $\mathbf{r}_{p,t}$ as the residual error between the observed measurements and the prediction in the $t^{\mathrm{th}}$ iteration, 
\begin{equation}
    \mathbf{r}_{p,t}=\mathbf{y}_p-e^{\mathsf{j}\hat{\phi}_{p,t}}\mathbf{A}_p\hat{\mathbf{x}}_{t}.
    \label{eq:residual}
\end{equation}
In the $t^{\mathrm{th}}$ iteration, matching pursuit algorithms in standard CS identify the column of $\mathbf{A}_p$ that results in the largest $|\mathbf{A}_{p}(:,k)^{\ast}\mathbf{r}^{(t-1)}_p|$, i.e., the absolute value of the correlation with the residue. In our problem, however, we have $P$ different residues derived from $P$ packets. As the measurements across the $P$ packets are non-coherent, our algorithm sums up the absolute values of the correlations and selects the index that maximizes this summation. The new element added to the support set is
\begin{align}
\label{eq:support_element_added}
    \hat{k}_t &= \underset{k \in [N] \setminus \Lambda_{t-1}}{\operatorname{arg max}}\sum_{p=1}^P|\mathbf{A}_{p}(:,k)^{\ast}\mathbf{r}_{p,t-1}|,
\end{align}
and the augmented support set is $\Lambda_t = \Lambda_{t-1}\cup\hat{k}_t$. For the special case when $P=1$, we observe that the objective in \eqref{eq:support_element_added} becomes identical to that used in matching pursuit algorithms for standard CS. In the $t^{\mathrm{th}}$ iteration, our approach explicitly excludes support elements in $\Lambda_{t-1}$ while the OMP inherently avoids selecting such elements. This is because the residue in our approach may not be orthogonal to the selected columns of the CS matrix, unlike the OMP.

\par After the support estimation step, we use an alternating minimization approach to estimate the non-zero entries of $\mathbf{x}$ over the identified support, and the phase errors $\{{\phi}_p \}_{p=1}^P$. Our method minimizes the sum of the squared errors across all the $P$ packets to estimate these quantities, 
\begin{equation}
    \label{eq:ls2}
    \hat{\mathbf{x}}_{\Lambda_{t}}, \hat{\boldsymbol{\phi}}_t = \underset{\mathbf{z}_{\Lambda_{t}},\boldsymbol{\delta}}{\operatorname{arg min}}
     \sum_{p=1}^P\left\|\mathbf{y}_{p} - e^{\mathsf{j}\delta_p}\mathbf{A}_{p,{\Lambda_{t}}}{\mathbf{z}}_{\Lambda_{t}}\right\|^2,
\end{equation}
where $\mathbf{A}_{p,{\Lambda_{t}}}$ is a submatrix of $\mathbf{A}_p$ obtained by retaining only those columns with indices in $\Lambda_{t}$. We observe that \eqref{eq:ls2} is a non-convex problem. It is, however, a standard least squares problem in $\mathbf{z}_{\Lambda_{t}}$ for a fixed $\boldsymbol{\delta}$. Furthermore, a closed form solution for $\boldsymbol{\delta}$ can be obtained for a fixed $\mathbf{z}_{\Lambda_{t}}$. Our alternating minimization procedure leverages both of these properties. 


\par We now provide closed-form expressions for phase recovery and signal estimation in the $\ell^{\mathrm{th}}$ iteration of alternating minimization in \eqref{eq:ls2}. For a fixed $\mathbf{z}^{(\ell-1)}_{\Lambda_t}$ in \eqref{eq:ls2}, we observe that the minimization problem is separable in $\{\delta_p\}_{p=1}^P$. The phase recovery problem for the $p^{\mathrm{th}}$ packet is then
\begin{align}
\nonumber
    \hat{{\phi}}^{(\ell)}_{p,t} &= \underset{\delta_p}{\operatorname{arg min}}\|\mathbf{y}_{p} - e^{\mathsf{j}\delta_p}\mathbf{A}_{p,\Lambda_{t}}\mathbf{z}^{(\ell-1)}_{\Lambda_t}\|_2\\
    \nonumber
    &=\underset{\delta_p}{\operatorname{arg max}} \text{{ Re}}\{e^{\mathsf{j}\delta_p}\mathbf{y}_p^{\ast}\mathbf{A}_{p,\Lambda_{t}}\mathbf{z}^{(\ell-1)}_{\Lambda_t} \}\\
    \label{eq:optim_phase}
    &=-\mathrm{phase}\left( \mathbf{y}_p^* \mathbf{A}_{p,\Lambda_{t}} \mathbf{z}^{(\ell-1)}_{\Lambda_t}\right).
\end{align}
After estimating $\hat{{\phi}}^{(\ell)}_{p,t}$ for each $p$, the signal in \eqref{eq:ls2} is estimated by solving a least squares problem. The solution is obtained by setting the gradient of the objective
to $0$ and is given by
\begin{equation*}    
\hat{\mathbf{x}}^{(\ell)}_{\Lambda_t}=\left(\sum_{p=1}^P\mathbf{A}^{\ast}_{p,\Lambda_{t}}\mathbf{A}_{p,\Lambda_{t}} \right)^{-1} \sum_{p=1}^Pe^{-\mathsf{j}\hat{{\phi}}^{(\ell)}_{p,t}}\mathbf{A}^{\ast}_{p,\Lambda_{t}}\mathbf{y}_{p}.
\end{equation*}
After the alternating minimization procedure converges in $L$ iterations, our method sets $\hat{\mathbf{x}}_{\Lambda_{t}}=\hat{\mathbf{x}}^{(L)}_{\Lambda_t}$ and $\hat{\boldsymbol{\phi}}_t=\hat{{\phi}}^{(L)}_{p,t}$. The subsequent PC-MP step updates the residue according to \eqref{eq:residual} and then identifies the next element of the support. A summary of our PC-MP technique is provided in Algorithm~\ref{alg:PC-MP}. 
\begin{algorithm}[h]
\caption{Proposed PC-MP algorithm for sparse recovery}
\label{alg:PC-MP}
\begin{algorithmic}[1]
\Statex \textit {\bf Input:} Partially coherent measurements $\{\mathbf{y}_p\}^P_{p=1}$, CS matrices $\{\mathbf{A}_p\}_{p=1}^P$, sparsity level $K$
\Statex \textit{\bf Initilize:} $\mathbf{r}_{p,0}=\mathbf{y}_p, \hat{\mathbf{x}}_0 = \mathbf{0}, \Lambda_0=\emptyset$
\For {$ t= 1,2,\ldots,K$}
\Statex \quad \; \textit{\#Support detection:}
\State $\hat{k}_{t}=\underset{k \in [N] \setminus \Lambda_{t-1}}{\operatorname{argmax}}\sum_{p=1}^P\left|\mathbf{A}_{p}(:,k)^{\ast} \mathbf{r}_{p,t-1}\right|$
\State $\Lambda_{t}=\Lambda_{t-1} \cup \hat{k}_{t}$
\State \textit{\#Signal estimation:} 
\While {$\left|\hat{\mathbf{x}}_{\Lambda_{t}}^{(l)}-\hat{\mathbf{x}}_{\Lambda_{t}}^{(l-1)}\right|>\epsilon$ or $\ell<L$}
\State $\hat{\phi}_{p,t}^{(l)} =-\angle \mathbf{y}_p^{H} \mathbf{A}_{p,\Lambda}\hat{\mathbf{x}}^{(l-1)}_{\Lambda_{t}}\; \forall p$ 
\State $ \hat{\mathbf{x}}^{(l)}_{\Lambda_{t}}\!=\!\left[\sum_{p=1}^P\!\mathbf{A}^*_{p,\Lambda_{t}}\!\mathbf{A}_{p,\Lambda_{t}} \right]^\dagger \!\!\left[\sum_{p=1}^P\!e^{-\mathsf{j}\phi_{p,t}^{(l)}}\!\mathbf{A}^*_{p,\Lambda_{t}}\!\mathbf{y}_{p} \right]$ 
\EndWhile
\State $\mathbf{r}_{p,t} = \mathbf{y}_{p}-e^{\mathsf{j}\hat{\phi}_{p,t}^{(l)}}\mathbf{A}_{p,{\Lambda_{t}}}\hat{\mathbf{x}}_{\Lambda_{t}}\;\forall p$
\EndFor
\State \textit {\bf Output:} $\hat{\mathbf{x}}$
\end{algorithmic}
\end{algorithm}
\par Now, we derive a coherence-based guarantee to successfully identify one element of the true support set of $\mathbf{x}$ using PC-MP, for real-valued CS matrices.   
\begin{prop}\label{prop:DetectionGuarantee}
Consider the channel estimation problem in \eqref{eq:measurement_model_pk} using measurements from $P$ packets. Assume that $\mathbf{A}_p$ is real-valued with normalized columns, i.e., $\|\mathbf{A}_p(:,i)\|_2=1$ for each $p$. For a constant $\beta>0$, the support element identified in the first iteration of PC-MP is correct 
with probability exceeding
\begin{equation}
1-\frac{1}{N^\beta P^\beta \sqrt{\pi(1+\beta) \log (NP)}},
\label{eq:prob}
\end{equation}
under the condition
\begin{equation}
|x_{\mathrm{max}}|\left(1-\frac{(2K-1)}{P}\sum_{p=1}^P\mu_p\right)>2\sigma\sqrt{2(1+\beta)\log (NP)},
\label{eq:condition}
\end{equation}
where $|x_{max}| = \underset{i\in\Lambda}{\operatorname{max}}\left\{|x_i|,  x_i \in \mathbf{x}\right\}$,
  $\mu_p$ is the coherence of $\mathbf{A}_p$~\cite{davenport2012introduction}, $\sigma$ is the standard deviation of Gaussian additive noise $\mathbf{w}$, and $K$ is the sparsity level of $\mathbf{x}$. 
  \end{prop}
  \begin{proof}
      See the appendix.
  \end{proof}
Intuitively, the result shows that support detection can only be successful under the assumption that the maximum absolute entry of $\mathbf{x}$ is ``larger'' than the additive noise level as described by \eqref{eq:condition}. Further, it also shows that a sufficient condition for successful identification of one element of the support in PC-MP's first iteration is $\sum_{p=1}^P\mu_p / P <1/(2K-1)$. Also, a good choice for the CS matrices for PC-MP is one with the lowest average coherence, i.e., $\sum^P_{p=1} \mu_p / P$. 

\begin{figure*}[htbp]
	\centering
	\subfigure[]{
			\includegraphics[trim=0.8cm 0cm 0.8cm 0.68cm, width=0.315\textwidth]{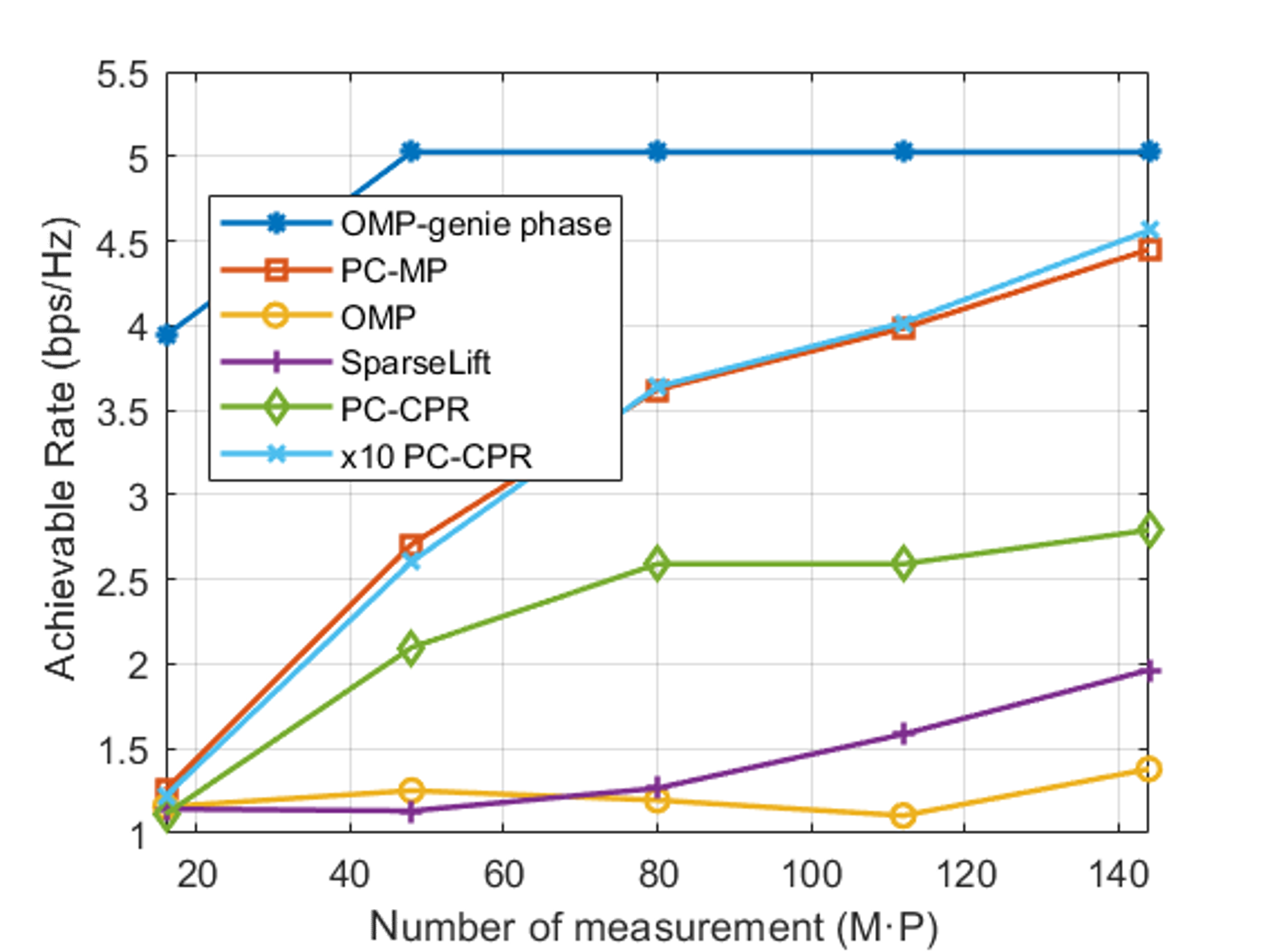}
		\label{fig:ach-mes}
	}
        \subfigure[]{
   		 	\includegraphics[trim=0.8cm 0cm 0.68cm 0.8cm,width=0.315\textwidth]{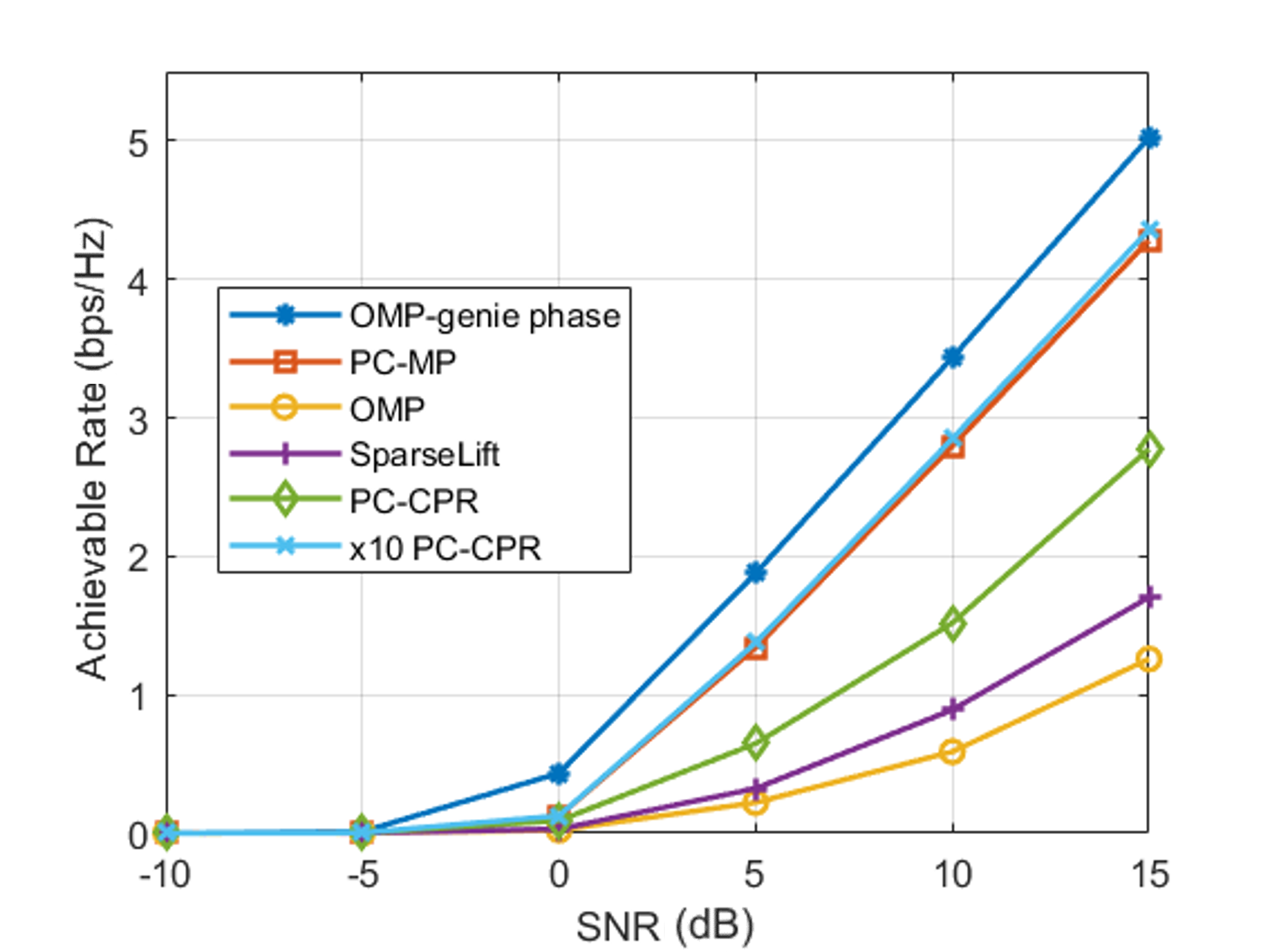}
	\label{fig:ach-snr}
        }
        \subfigure[]{
   		 	\includegraphics[trim=0.8cm 0cm 0.68cm 0.8cm,width=0.315\textwidth]{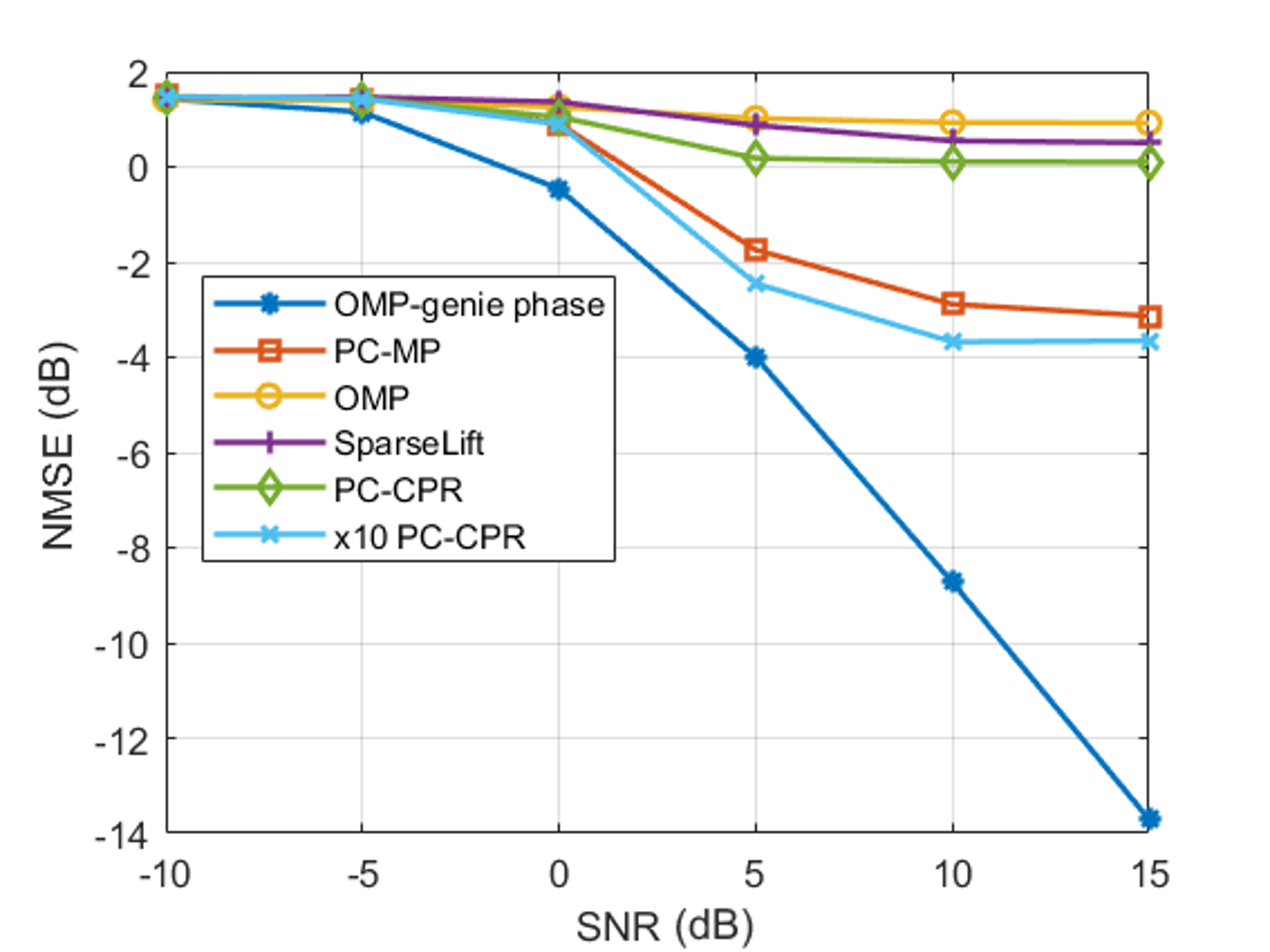}
	\label{fig:nmse-snr}
        }
	\caption{Simulation results that compare PC-MP against the benchmarks for $N=256$ antennas at the TX and $K=4$ sparse channels. (a) Achievable rate with the total number of measurements ($MP$) at SNR = $15$ dB. (b) Achievable rates with SNR for $MP=128$. (c) NMSE in the channel estimate with SNR for $MP=128$. We consider $M=16$ measurements per packet in all the plots. The number of iterations used for PC-CPR is same as that used for PC-MP, while $\times 10$ PC-CPR uses $10 \times$ more iterations than PC-MP. We observe that PC-MP outperforms the benchmarks for a fixed computational complexity.}
	\label{fig: comparison}
\end{figure*}
\section{Simulations}
\label{sec:simulations}
In this section, we compare the performance of our proposed method with OMP, self-calibration-based CS called \emph{Sparse-Lift}~\cite{ling2015self}, and partially coherent CS (PC-CPR)~\cite{9007529}, considering the system model described in section \ref{sec:model}.
\subsection{Benchmarking Algorithms}
\subsubsection{Reconstruction using Sparse-Lift}
The idea in Sparse-Lift is to jointly estimate the calibration errors and the sparse vector by solving for a high-dimensional lifted matrix, which is an outer product of the error vector and the sparse vector. 
To apply Sparse-Lift, we define the calibration error vector as $\mathbf{p}=(e^{\mathsf{j}\phi_1}, e^{\mathsf{j}\phi_2}, \cdots, e^{\mathsf{j}\phi_P})$ and the sparse vector is $\mathbf{x}$. The measurement vector $\mathbf{y}$ in \eqref{eq:measurement_modelMP} is then
\begin{equation}
\label{eq:Sparse-Lift}
    \mathbf{y} = \operatorname{diag}(\mathbf{B}\mathbf{p})\mathbf{Ax}+\mathbf{w}
\end{equation}
where $\mathbf{B}=\mathbf{I}_{P} \otimes \mathbf{1}^{\mathsf{T}}_M$. The $m^{\mathrm{th}}$ entry of $\mathbf{y}$ is given by
\begin{equation}
\label{eq:Sparse-Lift_measm}
    y[m] = \tilde{\mathbf{b}}_m^{\mathsf{T}} \mathbf{p} \mathbf{x}^{\mathsf{T}}\tilde{\mathbf{a}}_m+w[m],
\end{equation}
where $\tilde{\mathbf{b}}_m^{\mathsf{T}}$ is the $m^{\mathrm{th}}$ row of $\mathbf{B}$ and $\tilde{\mathbf{a}}_m$ is the transpose of the $m^{\mathrm{th}}$ row of $\mathbf{A}$. We define $\mathbf{X} = \mathbf{p}\mathbf{x}^{\mathsf{T}}$ as the lifted matrix, which is sparse as $\mathbf{x}$ is sparse. Using the identity $\mathbf{b}^{\mathsf{T}}\mathbf{X}\mathbf{a}=(\mathbf{a}^{\mathsf{T}} \otimes \mathbf{b}) \mathrm{vec}(\mathbf{X})$, we can rewrite \eqref{eq:Sparse-Lift_measm} as
\begin{equation}
y[m] =  (\tilde{\mathbf{a}}^{\mathsf{T}}_m \otimes \tilde{\mathbf{b}}_m^{\mathsf{T}}) \mathrm{vec}(\mathbf{X}) +w[m],
\end{equation}
which is a linear measurement of the lifted vector $\mathrm{vec}(\mathbf{X})$. With Sparse-Lift, $MP$ measurements of this form are used to solve for the sparse matrix $\mathbf{X}=\mathbf{p}\mathbf{x}^{\mathsf{T}}$. Then, the singular value decomposition of the estimate $\hat{\mathbf{X}} = \mathbf{U\Gamma V}^*$ is computed. The estimate of the sparse channel, up to a global scaling, is then $\hat{\mathbf{x}}=\overline{\mathbf{v}}_1$, where $\overline{\mathbf{v}}_1$ is the conjugate of the singular vector corresponding to the largest singular value.

\subsubsection{Partially Coherent Compressive Phase Retrieval (PC-CPR) }
\par The partially coherent CS model in~\cite{9007529} assumes that the coherent measurements are acquired using different radio frequency  chains. Although we consider a single radio frequency chain, our mathematical model is identical to the one in \cite{9007529}. The algorithm in \cite{9007529} is a two-stage method. In the first stage, the support set of the sparse vector is found by taking indices corresponding to the $K$ largest values of $\mathbf{z}$, where $z[k]=\sum_{p=1}^P|{\mathbf{A}_p(:,k)}^{\ast}\mathbf{y}_p|^2/M$. Then, the sparse estimate $\hat{\mathbf{x}}$ is initialized using the eigen-decomposition-based method in \cite{8424911}. In the second stage, the algorithm iteratively estimates the phase errors and the sparse signal. A hard thresholding algorithm is used to estimate the sparse signal. 
\subsection{Results and Discussion}
\par We consider a uniform linear array of size $N=256$ at the TX. The standard deviation of phase noise, conditioned on the previous sample, is given as $\tau = 2\pi f_c \sqrt{cT_s}$, with $f_c = 60$ GHz as the carrier frequency, $c = 4.7\times10^{-18}\left(\text{rad}\cdot \text{Hz}\right)^{-1}$\cite{9103070}, and $T_s=128$ ns as the time duration of a single measurement in a packet. When computing the variance associated with the first sample of each packet, we use $T_s=44 \, \mu$s. In our simulations, $\mathbf{x}$ is exactly sparse with $K=4$ non-zero components drawn from a Gaussian distribution. We use random circularly shifted Zadoff-Chu sequences for the beamformers $\{\mathbf{f}_m\}^M_{m=1}$ \cite{8306126}. With this setup, $M=16$ spatial channel measurements are acquired in each of the $P$ packets.
\par For the performance evaluation, we use the achievable rate as a metric and analyze how it varies with the SNR and the number of measurements. We define the SNR as $10\log_{10}\left(1/\sigma^2\right)$ and the achievable rate as $R = \log_2\left(1+|\mathbf{f}_{\mathrm{est}}^*\mathbf{h}|^2/{\sigma^2}\right)$, where $\mathbf{f}_{\mathrm{est}}$ is set to a unit-norm conjugate beamformer, i.e., $\mathbf{f}_{\mathrm{est}}=\overline{\mathbf{h}}_{\mathrm{est}}/ \Vert \mathbf{h}_{\mathrm{est}} \Vert_2$. As the algorithms can estimate the channel only upto a global phase, we define the normalized mean squared error as $\mathbb{E}[\mathrm{argmin}_{\delta}\Vert e^{\mathsf{j} \delta}\mathbf{h}_{\mathrm{est}}- \mathbf{h}\Vert^2_2]/\mathbb{E}[\Vert \mathbf{h}\Vert^2_2]$. Our results are summarized in Fig.~\ref{fig: comparison}.
\par From the plots in Fig. \ref{fig:ach-mes}, \ref{fig:ach-snr}, and \ref{fig:nmse-snr}, we observe that standard OMP that works well with known phase errors (genie phase) breaks down under unknown practical phase noise. Next, we notice that the proposed PC-MP algorithm outperforms Sparse-Lift in terms of NMSE and the rate. This is because our algorithm exploits the constant magnitude structure in the calibration vector, unlike Sparse-Lift. Furthermore, we also observed that Sparse-Lift required substantially higher computation time than our approach, as it solves for a high dimensional lifted vector with $NP$ variables. Our approach only solves for $N+P$ variables.  Finally, we observe that the proposed PC-MP algorithm outperforms PC-CPR for the same $K$ iterations. This is likely due to the nature of the algorithms, i.e., PC-CPR is based on hard thresholding while our approach is based on matching pursuit. Both these algorithms assume a known sparsity level of $K$. We would like to mention that the performance of PC-CPR for a large number of iterations ($\times 10$ PC-CPR) is close to our PC-MP method for $K$ iterations. For instance, PC-CPR required about $10K$ iterations ($25.4 \,\mathrm{ms})$ to achieve comparable performance as PC-MP for $K$ iterations ($5.4\, \mathrm{ms}$). These timings were observed in Matlab on a desktop computer. In summary, PC-MP can provide good channel estimates at a lower complexity than the benchmarks.

\section{Conclusions and Future Work}
In this paper, we introduced a greedy algorithm called PC-MP for mmWave channel estimation under phase noise. Our approach makes use of the partially coherent structure in the phase perturbed measurements to perform sparse recovery and phase error estimation. PC-MP performs support detection by considering measurements from different packets, and iteratively estimates the sparse vector through alternating minimization. A comprehensive comparison revealed that PC-MP outperforms benchmark algorithms, yielding higher achievable rates and lower NMSE. We also derived guarantees on PC-MP's performance in identifying one element from the true support. In future, we will extend our algorithm to account for off-grid effects and also analyze how well PC-MP can identify the entire support of the sparse vector. 

\section*{Appendix: Proof of Proposition~\ref{prop:DetectionGuarantee}}

We extend the support identification guarantee in~\cite{5483095} to the partially coherent case. Let $\Lambda$ be the true support of the sparse vector $\mathbf{x}$. Our proof verifies that when \eqref{eq:condition} holds, 
\begin{equation}\label{eq:verifiedcon}
    \underset{i\in\Lambda}{\operatorname{max}}\sum_{p=1}^P\left|\mathbf{A}_{p}(:,i)^{\mathsf{T}}\mathbf{y}_p\right|
    > \underset{i\notin\Lambda}{\operatorname{max}}\sum_{p=1}^P\left|\mathbf{A}_{p}(:,i)^{\mathsf{T}}\mathbf{y}_p\right|,
\end{equation}
which matches PC-MP's detection rule to successfully identify one element of the support in its first iteration. To this end, we show that \eqref{eq:verifiedcon} holds under the event $\mathcal{D}$ defined as
\begin{equation}\label{eq:D_defn}
    \mathcal{D}=\left\{\underset{1\leq i \leq N}{\operatorname{max}} \;\underset{1\leq p \leq P}{\operatorname{max}}|\mathbf{A}_{p}(:,i)^{\mathsf{T}}\mathbf{w}_p|<\zeta\right\},
\end{equation}
where $\zeta = \sigma\sqrt{2(1+\beta)\log N}$ controls the success probability of event $\mathcal{D}$. As $\left\{\mathbf{A}_{p}(:,i)^{\mathsf{T}}\mathbf{w} \right\}_{i,p}$ is jointly Gaussian, we get
\begin{align*}
        \operatorname{Pr}\left\{\mathcal{D} \right\}\!&=\operatorname{Pr}\left\{\underset{i,p}{\operatorname{max}}|\mathbf{A}_{p}(:,i)^{\mathsf{T}}\mathbf{w}_p|<\zeta \right\}\\
        &\geq\prod_{i=1}^N \!\prod_{p=1}^P\!\operatorname{Pr}\left\{|\mathbf{A}_{p}(:,i)^{\mathsf{T}}\mathbf{w}_p|\!<\!\zeta \right\}   \!  = \!\left[ 1\!-\!2Q\left(\frac{\zeta}{\sigma} \right)\right]^{NP}\!\!,
    \end{align*}
    where $Q(x)=(1 / \sqrt{2 \pi}) \int_x^{\infty} e^{-z^2 / 2}dz$ is the Gaussian tail probability and the last step uses \cite[theorem 1]{sidak} for jointly Gaussian random variables. Further, since the Gaussian tail probability is bounded by 
$Q(x)\leq\frac{1}{x \sqrt{2 \pi}} e^{-x^2 / 2},$
we obtain
\begin{equation*}
    \operatorname{Pr}\left\{ \mathcal{D} \right\} \geq \left(1-\sqrt{\frac{2}{\pi}} \frac{\sigma}{\zeta}e^{\frac{-\zeta^2}{2  \sigma^2}}\right)^{NP}\geq 1-NP\sqrt{\frac{2}{\pi}} \frac{\sigma}{\zeta}e^{\frac{-\zeta^2}{2  \sigma^2}}.\\
\end{equation*}
Substituting $\zeta = \sigma\sqrt{2(1+\beta)\log (NP)}$ in the above relation shows that
the event $\mathcal{D}$ occurs with probability exceeding \eqref{eq:prob}.

Once the event $\mathcal{D}$ holds for $p=1,2,\ldots,P$, we can simplify the left-hand side of \eqref{eq:verifiedcon} as
\begin{align}
\nonumber
        \underset{i\notin\Lambda}{\operatorname{max}}\sum_{p=1}^P\left|\mathbf{A}_{p}(:,i)^{\mathsf{T}}\mathbf{y}_p\right|\\
        \nonumber
        &\hspace{-3cm}=\underset{i\notin\Lambda}{\operatorname{max}}\sum_{p=1}^P \left|\mathbf{A}_{p}(:,i)^{\mathsf{T}}\mathbf{w}+\sum_{k\in\Lambda} e^{\mathsf{j}\phi_p}\mathbf{A}_{p}(:,i)^{\mathsf{T}}\mathbf{A}_{p}(:,k){x}_k\right|\\
        \nonumber
        &\hspace{-3cm}\leq  \underset{i\notin\Lambda}{\operatorname{max}}\sum_{p=1}^P\left|\mathbf{A}_{p}(:,i)^{\mathsf{T}}\mathbf{w}\right| \\
        \nonumber
        &\hspace{-2.7cm}+\underset{i\notin\Lambda}{\operatorname{max}}\sum_{p=1}^P\sum_{k\in\Lambda}\left|e^{\mathsf{j}\phi_p}\mathbf{A}_{p}(:,i)^{\mathsf{T}}\mathbf{A}_{p}(:,k){x}_k\right|\\
        \label{eq:bound_out_support}
        &\hspace{-3cm}\leq P\zeta+K\sum_{p=1}^P\mu_p\left|x_{\mathrm{max}} \right|,
    \end{align}
    where the last step follows from \eqref{eq:D_defn} and the definition of $\mu_p$ and $x_{\mathrm{max}}$.
For a sparsity level $K\geq1$, using \eqref{eq:measurement_model_pk} and column normalization assumption on $\mathbf{A}_p$, we can bound the right-hand side of \eqref{eq:verifiedcon} as
    \begin{align}
    \nonumber
        \underset{i\in\Lambda}{\operatorname{max}}\sum_{p=1}^P\left|\mathbf{A}_{p}(:,i)^{\mathsf{T}}\mathbf{y}_p\right|\\
        \nonumber
&\hspace{-3cm}=\underset{i\in\Lambda}{\operatorname{max}}\sum_{p=1}^P \left|\mathbf{A}_{p}(:,i)^{\mathsf{T}}\left(\sum_{k\in\Lambda}e^{\mathsf{j}\phi_p}\mathbf{A}_{p}(:,k){x}_k+\mathbf{w}\right)\right|\\
\nonumber
&\hspace{-3cm}=\underset{i\in\Lambda}{\operatorname{max}}\sum_{p=1}^P |{x}_i|\\
\label{eq:bound_in_support_1}
&\hspace{-3cm}\;-\left|\sum_{k\in\Lambda\setminus\{i\}}e^{\mathsf{j}\phi_p}\mathbf{A}_{p}(:,i)^{\mathsf{T}}\mathbf{A}_{p}(:,k){x}_k+\mathbf{A}_{p}(:,i)^{\mathsf{T}}\mathbf{w}\right|
\end{align}
Here, from the definition of $x_{\mathrm{max}}$ and $\mu_p$, we get
\begin{align}
\notag
\underset{i\in\Lambda}{\operatorname{max}}\sum_{p=1}^P\!\left|\mathbf{A}_{p}(:,i)^{\mathsf{T}}\mathbf{y}_p\right|\!&\geq \! P|x_{\mathrm{max}}|-\!\sum_{p=1}^P \!\left[\sum_{k\in\Lambda/i} \mu_p|x_{max}| + \zeta\right]\\
  \label{eq:bound_in_support_2}
        &\geq \! P|x_{\mathrm{max}}|\!-\!P\zeta\!-\!(K\!-\!1)\!\sum_{p=1}^P\mu_p|x_{\mathrm{max}}|.
     \end{align}     
In the first iteration of PC-MP, support identification is successful when \eqref{eq:verifiedcon} holds. We observe that the condition in  \eqref{eq:verifiedcon} holds if the lower bound in \eqref{eq:bound_in_support_2} exceeds the upper bound in \eqref{eq:bound_out_support}, which is equivalent to
\begin{equation}
P|x_{\mathrm{max}}|-2P\zeta-(2K-1)\sum_{p=1}^P\mu_p\left|x_{\mathrm{max}}\right|>0,    
\end{equation}
which is the condition stated in \eqref{eq:condition}. 
\par Our guarantees are to identify only one element of the support set and not the entire support. This limitation is because of the residual errors arising due to joint phase recovery and signal estimation. These errors may not be orthogonal to the columns already selected, which prevents us from applying the induction argument discussed for the coherent case in~\cite{5483095}. 

\bibliographystyle{ieeetr}
\bibliography{ref}

\vspace{12pt}

\end{document}